\documentclass[aps, twocolumn,amssymb,amsmath,amsfonts,superscriptaddress, showpacs]{revtex4}

\usepackage{graphicx}
\usepackage{bm,bbm}
\usepackage{epstopdf}
\usepackage{amsthm}
\usepackage{amsmath}
\usepackage{color}
\usepackage{float}

\usepackage{hyperref}
\usepackage{graphicx}
\usepackage{graphics}
\usepackage{physics}

\newcommand{\rd}[1]{{\color{red}#1}}

\usepackage[T1]{fontenc}
\usepackage[english]{babel}

\providecommand{\abs}[1]{|#1|}

\newtheorem{prop}{PROPOSITION}

\newtheorem{obs}{OBSERVATION}
\newtheorem{conj}{CONJECTURE}

\begin{document}
\title{Trade--off relations for operation entropy of complementary quantum channels}

\author{Jakub Czartowski}
\email{jakub.czartowski@student.uj.edu.pl} 
\affiliation{Smoluchowski Institute of Physics, 
Jagiellonian University, ul. {\L}ojasiewicza 11,  30-348 Krak{\'o}w, Poland}

\author{Daniel Braun}
\affiliation{University of Tuebingen, Germany}

\author{Karol {\.Z}yczkowski}
\affiliation{Smoluchowski Institute of Physics, 
Jagiellonian University, ul. {\L}ojasiewicza 11,  30-348 Krak{\'o}w, Poland}
\affiliation{Center for Theoretical Physics, Polish Academy of Sciences, 
Al. Lotnik{\'o}w 32/46, 02-668 Warsaw, Poland}

\date{August 9, 2019}

\begin{abstract}

  The entropy of a quantum operation, defined as the von Neumann entropy of the corresponding
Choi-Jamio{\l}kowski state, characterizes the coupling of the principal system with
the environment. For any quantum channel $\Phi$ acting on
a state of size $N$ one defines the complementary channel
$\tilde \Phi$, which sends the input state into the state of the environment after the operation.
Making use of subadditivity of entropy we show that for any dimension $N$
the sum of  both entropies,  $S(\Phi)+ S(\tilde \Phi)$, is bounded from below.
This result characterizes the trade-off between the information on the initial quantum state
accessible to the principal system and the information leaking to the environment. 
For  one qubit maps, $N=2$, we describe the interpolating family of depolarising maps, 
for which the sum of both entropies gives the lower boundary of the region allowed
in the space spanned by both entropies.
\end{abstract}


%\pacs{03.65.Aa}
\maketitle

%%%%%%%%%%%%%%%%%%%%%%%%%%%%%%%%%%%%%%%%%%%%%%%%%%%%%%%%%%%%%%%%%%%%%%%%%%%%%%%%
\section{Introduction}
%%%%%%%%%%%%%%%%%%%%%%%%%%%%%%%%%%%%%%%%%%%%%%%%%%%%%%%%%%%%%%%%%%%%%%%%%%%%%%%%	

Any
time evolution
of a density matrix $\rho$ can be described by a quantum operation 
$\Phi$, often called a {\sl quantum channel}. It is defined by a completely positive,
 trace preserving linear map, which sends the set of all quantum
 states into itself \cite{Pe95}. 
 Such a channel can be considered as a generalization of the unitary
 evolution of a density matrix
 that takes into account the interaction of the system with an environment or with the measurement apparatus. 
 The action of such a map, following Stinespring's dilation theorem, can also be interpreted as a unitary evolution of the joint system composed of the principal system and the environment, followed by the partial trace over the environment. 
 
For any quantum operation $\Phi$ one defines the complementary
operation $\tilde \Phi$, which maps the initial state $\rho$ into the final state of the 
environment \cite{H05}. In the language of quantum communication, 
the state $\rho'=\Phi(\rho)$ describes 
the final state at the output of the channel, 
while the state $\rho''={\tilde\Phi}(\rho)$
describes the final state of the eavesdropper, 
who attempts to intercept the information transmitted
in the state $\rho$.

To quantify the amount of information encoded in a classical or
quantum state, various entropic measures (e.g. those based on von Neumann
entropy) are often invoked \cite{OP}.
A similar entropic approach can also be used to describe
the information flow induced by a quantum channel.
In particular, the notions of {\sl Holevo quantity} \cite{Ho73},
coherent information and information exchange \cite{SN96,HG12}
defined for a map  $\Phi$ and an initial state $\rho$
are based on the von Neumann entropy.

To analyze the set of quantum operations it is convenient 
to make use of the known Choi-Jamio{\l}kowski isomorphism \cite{Ja72,Ch75},
which relates a quantum operation $\Phi$ acting on an $N$ dimensional state $\rho$
to an auxiliary state $\sigma_{\Phi}$ defined on an extended space of size $N^2$.
If this state is pure the corresponding map is unitary, $\Phi_U(\rho)=U\rho U^{\dagger}$,
while the maximally mixed state
for  $\sigma_\Phi=\rho^*$
corresponds to the totally depolarizing channel.
Thus the degree of mixing of the Choi-Jamio{\l}kowski state $\sigma_{\Phi}$,
characterized by its von Neumann entropy, can be used to describe the   
degree of nonunitarity of the map $\Phi$ and the coupling with an environment. 

More formally, for any channel $\Phi$ one defines its entropy \cite{ZB04,BZ06}
as the von Neumann entropy of the corresponding Choi-Jamio{\l}kowski state,
 $S(\Phi):=S(\sigma_{\Phi})$. This quantity, also called the {\sl entropy of an operation}
 or {\sl map entropy} yields an upper bound  \cite{RFZ10}
  for the Holevo quantity $\chi$  \cite{Ho73},  associated with the
  transformation of the maximally mixed state $\rho_*$  by the operation $\Phi$.
  
The
  entropy of an operation is additive with respect to the tensor product \cite{RFZ11},
$S(\Phi \otimes \Psi)=S(\Phi)+S(\Psi)$.
It is also known that for bistochastic channels, which preserve the maximally
mixed state, the map entropy is subadditive with respect to concatenation \cite{RFZ08}.
Furthermore, the entropy of an operation
satisfies a trade--off relation  \cite{RPRZ13}
with respect to the receiver entropy, which depends on singular values of the 
superoperator $\Phi$ and describes the
receiver's
knowledge of the output state without any
information on the input.

The aim of this work is to extend these results to establish
a trade o ff relation concerning the operation entropies
of a given quantum channel and its complementary.  The obtained lower bound for the sum of both entropies,
$S(\Phi)+S({\tilde \Phi})$,
is  valid for an arbitrary system size $N$.
We also analyze distinguished  channels, 
for which the above sum 
attains mimimal values and yields
the lower boundary of the allowed
set in the plane $\bigl(S(\Phi), S({\tilde \Phi}) \bigr)$.
In the case of one-qubit maps we
identify the corresponding family of depolarising channels and provide proof of extremity. Furthermore, we present families of channels in product dimensions saturating the obtained general bound and give conjecture concerning method of obtaining precise boundary of the allowed set of entropies in general dimension $N$.

\section{Setting the scene: quantum channels and their entropies}

A quantum channel $\Phi$ denotes a trace preserving and completely positive linear map
which maps a quantum state $\rho\in \mathcal{M}_N$ to another state of a possibly different dimension $M$, namely 
$\rho'=\Phi\qty(\rho)\in\mathcal{M}_M$. 
Any such
channel, also called a quantum operation,
can be represented by a
set
of $m$ Kraus operators $K_i$
\begin{equation}
    \Phi\qty(\rho) = \sum_{i=1}^m K_i \rho \qty(K_i)^\dag ,  
        \label{Krauschanneldef}
\end{equation}
and this representation is not unique.
In general the number $m$ is arbitrary, but for any channel $\Phi$ one can find the
canonical representation for which $m\leq N^2$ - see e.g. \cite{BZ06}.
Any map of the above form is completely positive, but to satisfy the trace preserving condition,
${\rm Tr}\Phi(\rho)= {\rm Tr}\rho$,
the Kraus operators have to fulfil the
identity resolution condition
$ \sum_i \qty(K_i)^\dag K_i = \mathbb{I}_N. $

Let us introduce the maximally entangled state on the extended system of dimensionality $N\times N$,
$\ket{\psi_+} = \frac{1}{\sqrt{N}}\sum_{i=1}^N\ket{i}\otimes\ket{i}$. For any map $\Phi$ taking states from $\mathcal{M}_N$ to states of the same dimensionality it allows us to define the corresponding Choi-Jamio{\l}kowski state  \cite{Ja72},
obtained by action of an extended channel, ${\mathbbm I} \otimes \Phi$,
on the maximally entangled state, which can be written as a block matrix %\qq of size
of linear dimension $N^2$, consisting of $N$ columns of $N$ $N\times
N$ blocks,   
\begin{equation}
   %  \frac{D_\Phi}{N}
    \sigma_{\Phi}
    : = \qty({\mathbbm I}\otimes\Phi)\qty(\op{\psi_+}) = \mqty(\Phi\qty(\op{1}) & \hdots & \Phi\qty(\op{1}{N}) \\ \vdots & \ddots & \vdots \\ \Phi\qty(\op{N}{1}) & \hdots & \Phi\qty(\op{N})). 
    \label{choijamstate}
\end{equation}

Making use of this isomorphism one defines \cite{ZB04} the entropy of the channel 
as the von Neumann entropy $S(\rho) = -\tr \rho\log\rho$ of the corresponding Choi-Jamio{\l}kowski state,
\begin{equation}
    S^{\text{map}}\qty(\Phi) = S\qty(\sigma_\Phi). 
       \label{mapentropydef}
\end{equation}
Since the Choi-Jamio{\l}kowski 
state has dimension $N^2$,
the entropy of a channel is bounded from above,
%\begin{equation}
$
    S^{\text{map}}\qty(\Phi) \leq 2 \log N.
    $ %\label{mapentropylimit}
%\end{equation}
The entropy of a unitary channel is equal to zero, 
while the upper bound is saturated for the maximally depolarizing channel  \cite{BZ06}.

If two quantum channels are close, so that the trace distance between the corresponding 
Choi-Jamio{\l}kowski states is small, then due to the Fannes theorem \cite{Fa73}
 the entropies of both channels are similar.
 The entropy of a channel is easier to determine than other entropic quantities,
 like the minimal output entropy, as no minimization is involved.

Several  interesting properties of the entropy of a channel were obtained during the recent decade
\cite{RFZ08,RFZ10,RFZ11,RPRZ13}.
However, to avoid misunderstanding it is worth
mentioning here that recently the
notion of {\sl entropy of a channel} was used in
a similar spirit for a related
 but different quantity  \cite{gour}, calculation of which requires optimization.

%\subsection{Unitary representation of a channel and complementary channel definition}

Any quantum operation  $\Phi$ can
also be represented in an environmental form using Stinespring dilation theorem, 
so that the initial state $\rho$ is coupled to an environment of
dimension $M$ by a unitary operator $U$
acting on the combined Hilbert space of 
dimension $NM$,
\begin{equation}
    \Phi\qty(\rho) = \Tr_{\mathcal{E}} \qty[U\qty(\rho\otimes\op{1})U^\dag]. \label{unitarydynamicsdef}
\end{equation}
It can be assumed that the  $M$--dimensional environment ${\mathcal{E}}$
is intially prepared in an arbitrary pure state  $\omega = \op{1}$.
The above expression is equivalent to the Kraus form  (\ref{Krauschanneldef}),
as the Kraus operators $K^i$ are determined by the  block-column of  the matrix $U$,
namely $ K^i_{jk} = U_{j+(i-1)N\,k}$,
 and  its  unitarity imposes the trace preserving condition.
The  complementary operation $\tilde \Phi$ is defined \cite{H05}
by an analogous formula with partial trace over the principal system $S$,
%
%one defines the complementary
%operation $\tilde \Phi$, which maps the initial state $\rho$ into the final state of the 
%environment \cite{H05}. 
%and it is easily verified that for ancilla of dimension $M$ one can have at most $M$ non-zero Kraus operators. With such definition we may proceed to define complementary channel $\tilde\Phi$ as in \cite{SRZ16} as the action of the unitary dynamics on the environment state:
\begin{equation}
    \tilde\Phi\qty(\rho) = \Tr_S \qty[U\qty(\rho\otimes\op{1})U^\dag], 
    \label{complementarydef1}
\end{equation}
so it concerns the state of the environment after the operation. 

The complementary channel $\tilde \Phi$ 
can also be written  with use of the orthogonal SWAP operation, defined as $O_{SWAP}(\rho\otimes\sigma)O_{SWAP} = \sigma\otimes\rho$, which 
exchanges the principal system with the environment:
\begin{equation}
    \tilde\Phi\qty(\rho) = \Tr_{\mathcal{E}} \qty[O_{SWAP}U\qty(\rho\otimes\op{1})U^\dag O_{SWAP}]. \label{complementarydef2}
\end{equation}
These relations imply that the Kraus operators $\tilde K_i$ 
forming the complementary operation $\tilde \Phi$
can be obtained from Kraus operators $K_i$ corresponding to the 
original channel by exchanging the rows of matrices \cite{SRZ16},
\begin{equation}
    \qty(\tilde K_\alpha)_{ij} = \qty(K_i)_{\alpha j} 
    \label{Krausrowswap}
\end{equation}
where $i,j = 1,\hdots,N$, $\alpha = 1,\hdots,M$ and Kraus operators that are not specified are assumed to be equal to zero.
%\rd{
%As it is not essential for any part of further text, may be omitted or shifted to appendices
%
%Consider now a channel written in the environmental form,
%$\Phi(\rho) = \Tr_\mathcal{E}\qty[U \qty(\rho\otimes\sigma) U^\dag]$,
%followed by a unitary rotation, $ \Phi'\qty(\rho) = V\Phi\qty(\rho)V^\dag$.
%%
%%\begin{equation}
%%    \Phi'\qty(\rho) = V\Phi\qty(\rho)V^\dag
%%\end{equation}
%This concatenation can be represented by Kraus operators $G$,  
%obtained by combining the initial Kraus operators,
%\begin{equation}
%    G^\alpha_{ij} = \sum_\beta V^\alpha_\beta K^\beta_{ij}.
%\end{equation}
%Making use of  Eq. (\ref{Krausrowswap}) we see that 
% the complementary channel $\tilde{  \Phi '}$
%is determined by  operators $\tilde G^i$, 
%\begin{align}
%    \tilde G^i_{\alpha j}=  G^\alpha_{ij} & = \sum_\beta V^\alpha_\beta \tilde K^i_{\beta j}.
%   % \tilde G^\alpha & = V \tilde K^\alpha
%\end{align}
%In a more concise notation one has $ \tilde G^i  = (V \tilde K)^i$
%which shows that the unitary transformation performed on the output of a channel
%leaves its complementary counterpart with the same Kraus operators up to a unitary transformation.}

\section{Bounding the sum of two entropies}

In order to establish bounds for the entropy of a channel
we start by  pointing out  the relation between $S(\Phi)$ 
and the entropy of the image of the maximally mixed state,
$\rho_*={\mathbbm I}/N$, 
with respect to the complementary channel.
%
%For any channel $\mathcal{M}_N \ni \rho \rightarrow \Phi\qty(\rho) \in \mathcal{M}_N$ defined by unitary operation $U$ on $N$-dimensional system with $M$-dimensional ancillary system:
 %   \begin{equation}
   %     \Phi(\rho) = \Tr_\mathcal{E} U\qty(\rho \otimes \op{1})U^\dag
  %  \end{equation}
%we have the following relations:
    
\begin{prop}
\label{prop1}

Consider a channel $\Phi$ acting on an $N$ dimensional system by
coupling it with the environment of
dimension $M$.
The  entropy of  the channel is equal to the entropy of the image of
the 
maximally mixed state $\rho_* = \mathbb{I}/N$ 
under the complementary channel $\tilde\Phi$.
An analogous relation holds for the complementary channel, 
%Similarly map entropy for $\tilde\Phi$ is equal to entropy of image of $\rho_*$ under the channel $\Phi$
    \begin{align}
        S^{\text{map}}\qty(\Phi) & = S\qty(\tilde\Phi\qty(\rho_*)) & 
        S^{\text{map}}\qty(\tilde\Phi) & = S\qty(\Phi\qty(\rho_*)) \label{prop1eq1}
    \end{align}
\end{prop}

\begin{proof}
	The proof of this proposition is given in Appendix \ref{proofsofprops}.
\end{proof}

To characterize the quantum information remaining in the initial
state $\rho$ after an action of a given channel $\Phi$ 
one uses the {\sl coherent information} \cite{SN96} expressed by 
the 
difference between the 
output entropies of a channel and its complementary,
\begin{equation}
    I_{coh}\qty(\Phi,\rho) = S\qty(\Phi(\rho)) - S\qty(\Tilde\Phi(\rho)).
\end{equation}
Select now the initial state to be maximally mixed, $\rho=\rho_*$.
Due to Proposition \ref{prop1} the coherent information of $\Phi$
can be expressed in this case by the inversed difference 
of the entropy of the channel and its complementary, 
\begin{equation}
    I_{coh}\qty(\Phi,\rho_*) = 
      S^{\text{map}}\qty(\Tilde\Phi) - S^{\text{map}}\qty(\Phi).
\end{equation}
The more unitary 
the 
channel $\Phi$, the smaller its entropy 
$S^{\text{map}}\qty(\Phi)$, and the less information leaks
out of the initial state $\rho_*$ to the environment.

\smallskip

Proposition 1 allows us to demonstrate a general bound for 
the sum of entropies of a channel and its complementary.
\begin{obs}
\label{obs1}
    The entropies of a given channel  $\Phi$ and of its complementary $\tilde \Phi$  
    are bounded by the following inequality
    \begin{align}
        S^{\text{map}}\qty(\Phi) + S^{\text{map}}\qty(\tilde\Phi) = S^{\text{map}}\qty(\Phi\otimes\tilde\Phi) & \geq \log N \label{prop2ieq3}
        % \abs{I_{coh}(\Phi,\rho_*)}  & \leq \log N \label{prop2ieq4}
    \end{align}
\end{obs}

\begin{proof}
   Observation \ref{obs1} is easily proven by considering
  \eqref{prop1eq1} and Proposition 8 in \cite{RPRZ13}. An alternative
  proof is given in Appendix \ref{proofsofprops} 
\end{proof}

%\medskip
%{\bf Remarks  (KZ)}
%\smallskip 
%
%1. \rd {This conclusion can also be reached from Proposition \ref{prop1}
%by applying Proposition 8 from Ref. \cite{RPRZ13}. Agree.}
%\smallskip 
%
%2. \rd {Due to additivity of the entropy of an operation with respect to the tensor product \cite{RFZ11},
%relation  (\ref{prop2ieq3}) can be rewritten as 
%
%
%$$    S^{\text{map}}\qty(\Phi \otimes \tilde \Phi) \ge \ln N.$$
%
%
%On one hand this relation looks interesting and promissing. On the other hand, it is also possible
%that this simple formula is a direct consequence of some basic properties of any quantum operation $\Phi$,
%and our Proposition 2 might occur to be only another formulation of some earlier known results...
%A detailed study of this issue is necessary}.
%\smallskip 
%
%3. \rd{If the latter is the case we shall need to improve Proposition 2, e.g. by finding for $N=2$ the analitical form of the blue curve in Fig. 1, which gives the lower boundary  for the allowed set and to prove that it gives the lower bound and later by generalizing this statement for larger $N$.} DONE, results to be simplified and included in an appendix
%
%\medskip

\begin{obs}
  The following two inequalities hold for any channel $\Phi$ acting
  on a
  space of dimension $N$ with $M$-dimensional environment: 
    \begin{align}
        S^{\text{map}}\qty(\Phi) &\leq \log M,&{\rm and}&&
        S^{\text{map}}\qty(\tilde\Phi) &\leq \log N \label{obs1ieq1}.    
    \end{align}
\end{obs}

\begin{proof}
	Proof follows from direct inspection of \eqref{prop1eq1}.
\end{proof}

% \rd{\textit{System-environment asymmetry of the bounds should be emphasized at this point. Figures with qubit-qutrit and qutrit-qubit distributions should be placed here. Perhaps qubit-ququart figure should be placed here as well to emphasize degeneration to a point of bound \eqref{prop1ieq4} at certain dimensionality of the environment $M \geq N^2$}}

%The coherent information \cite{SN96}, transmitted through a channel $\Phi$ acting on the initial
%state $\rho$ reads  $I_{\rm coh}(\phi,\rho)=S(\Phi(\rho))-S({\tilde{\Phi}}(\rho))$.
%If the initial state is maximally mixed, $\rho_*={\mathbbm I}/N$,
%the coherent information is equal to the difference of both entropies, 
%$I_{\rm coh}(\phi,\rho_*)=S({\tilde{\Phi}}) - S(\Phi)$.

%\subsection{Saturation of lower bound for total entropy}
Let us consider a channel $\Phi_U$ on an $N$-dimensional system with
one Kraus operator $K_1 = U$ %\qq  being 
which is hence a unitary operator.
% \qq From here t
The complementary channel $\tilde\Phi_U$ is given by set of $N$ Kraus
operators $\tilde K_i$ determined by succesive rows of the Kraus operator
$K_1$,
$
    \Tilde{K}_i = \sum_{j=1}^N U_{ij} \op{1}{j}.
$
%The Choi-Jamio{\l}kowski state can be considered as its block form
%\eqref{choijamstate}. %\qq I don't know what you want to say with
% this.
For the unitary channel $\Phi_U$, the
Choi-Jamio{\l}kowski state %\qq do you have
                                %another way of defining it in mind,
                                %rather than the usual one?
has the same entropy as
the state related to the identity channel, thus 
   $S^{map}\qty(\Phi_U) = 0$.
The Choi-Jamio{\l}kowski state of the complementary channel $\Tilde{\Phi}$ is easily
found to be composed of blocks $\sigma_{\mu\nu} = \rho_*\otimes\op{1}$
 of dimension $N$, 
thus $
    S^{map}\qty(\Tilde{\Phi}_U) = \log N.
$
Collecting the two entropies, we can see that for unitary channels of
dimension $N$ the total entropy of channel and its complementary is
given by 
\begin{equation}
    S^{map}\qty(\Phi_U) + S^{map}\qty(\tilde\Phi_U) = \log N \label{eq13}
\end{equation}
and saturates the bound \eqref{prop2ieq3}.

Another way of calculating \eqref{eq13} 
is through use of Eqs. \eqref{prop1eq1}. A unitary channel acting on a maximally mixed state $\rho_* = \mathbb{I}/N$ leaves it unchanged, $
    \Phi_U\qty(\rho_*) = \rho_*$, from which $S\qty(\Phi\qty(\rho_*)) = \log N$.
Furthermore, the complementary channel takes it to a projector state of the environment
$    \tilde\Phi_U\qty(\rho_*) = \op{1}$ so that $S\qty(\tilde\Phi_U\qty(\rho_*)) = 0.
$

In the following sections we aim to provide even more insight about the structure of the set of allowed operations in the plane of entropy of a channel and of its complementary.

\section{Qubit channels}

The easiest system to consider is a qubit system coupled to a qubit environment.  In such a case the channel $\Phi$ and its complementary $\Tilde{\Phi}$ can be both represented by two Kraus operators. In terms of entropies' plane $\qty(S^{map}(\Phi),S^{map}(\Tilde{\Phi})) \equiv \qty(S, \Tilde{S})$, there are three  points in the boundary of available region $\mathcal{A}_2$, for which we identify representantive channels:

\begin{enumerate}
    \item $(0, \log 2)$: Unitary channels $\Phi_U$, for which an exemplary channel is the identity channel $K_1 = \mathbb{I}$;
    \item $(\log 2, 0)$ One-step emission channel $\Phi_{E}$ given by $K_1 = \op{1}$ and $K_2 = \op{1}{2}$, complementary to identity channel, which sends any density matrix into the ground state $\Phi(\rho) = \op{1}$;
    \item $(\log 2, \log 2)$: Coarse graining channel $\Phi_{CG} = \tilde \Phi_{CG}$, sending any quantum state into the diagonal matrix $\Phi(\rho) = {\rm diag}(\rho)$, determined by Kraus operators $K_1 = \op{1}$ and $K_2 = \op{2}$. 
\end{enumerate}

In Figure \ref{fig:qbitfullbound} we show the space available for
qubit-qubit channels together with its boundary. The upper boundary for
complementary channel entropy is given by a line segment $S = \ln 2$
and for any given value of $\Tilde{S}$ there exists a channel, given
by interpolation between coarse graining channel $\Phi_{CG}$ and
unitary channel $\Phi_{U}$. Similarly, the upper boundary for channel
entropy can be found by interpolating between $\Phi_{CG}$ and
spontaneous emission channel $\Phi_{E}$. The third curve comprising
the full boundary is given below. 

\begin{prop} \label{prop_qbit_bound}
	The lower boundary curve of the allowed set  $\mathcal{A}_2$
	of one-qubit channels represented in the entropy plane,
	 minimizing $S^{map}(\tilde\Phi)$ for a fixed channel entropy $S^{map}(\Phi) = const.$ is given parametrically by
	\begin{align} \label{boundary_qubit_eqn}
		p(a) = (-a\log a - (1-a)\log(1-a)\nonumber, \\ -(\frac{1}{2} - a) \log(\frac{1}{2} - a)-(\frac{1}{2} + a) \log(\frac{1}{2} + a))
	\end{align}
	{for $a\in[0,\frac{1}{2}]$.} 
\end{prop}

\begin{proof} A detailed proof of the extremity of this curve, 
labeled in Fig. \ref{fig:qbitfullbound}
 by the letter `c', 
 is given in Appendix \ref{prop3proof}. 
\end{proof}
 
The simplest way to obtain this parametric formula is to consider spontaneous emission channels interpolating between $\Phi_U$ and $\Phi_{E}$, given in terms of Kraus operators
\begin{align}
K_1 = \mqty(1 & 0 \\ 0 & \sqrt{x}) && K_2 = \mqty(0 & \sqrt{1-x}\\0 & 0)
\end{align}
with $x = 2a - 1$. Entropy of this channel and its complementary follow the extreme curve in \eqref{boundary_qubit_eqn}. It is important to note that these channels can be seen as those for which for constant entropy of the image of the maximally mixed state $\Phi(\rho_*)$ the information escaping to the environment is minimized. %\qqq  

\begin{figure}[h]
    \centering
    \includegraphics[width=.9\linewidth]{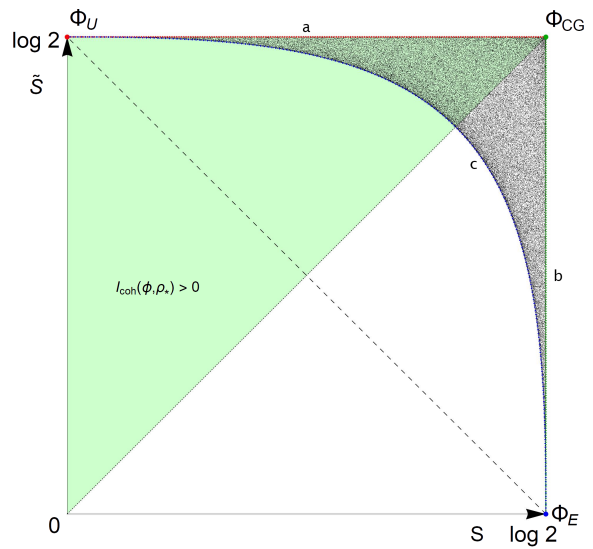}
    \caption{Scatter plot of a sample of randomly generated one qubit-qubit channels of rank two  in the entropy plane $(S=S(\Phi), 
    {\tilde S}=S({\tilde \Phi}))$,
     which belong to the allowed set  $\mathcal{A}_2$.
     Lines $a$ and $b$ saturate upper bounds \eqref{obs1ieq1} on the entropy of channel and its complementary, respectively, and are easily found by interpolating between coarse graining channel $\Phi_{CG}$ and spontaneous emission channel $\Phi_{E}$ or unitary channel $\Phi_U$. Minimal curve $c$, described by parametric formula \eqref{boundary_qubit_eqn}, can be found by interpolation between $\Phi_U$ and $\Phi_{E}$.}
    \label{fig:qbitfullbound}
\end{figure}

\section{General dimension}

\subsection{Emission channels and matrix $L$}

Let us introduce a left upper triangular matrix $L$ of dimension $N$ with entries $L_{ij} \in \qty{0,1}$ such that $\sum_{i=1}^N L_{i,\,N + 1 - i} = 1$. Any such matrix determines a valid quantum channel $\Phi$ with Kraus operators
\begin{equation}
  \label{eq:KfromL}
K_i = \sum_{j=1}^{N+1-i} L_{ij} \op{j}{j + i - 1},  
\end{equation}
with $i=1,\ldots,N$. %\qqq

Its complementary channel $\Tilde{\Phi}$ is found by similar formula $\Tilde{K}_i = \sum_{j=1}^{N+1-i} L^T_{ij} \op{j}{j + i - 1}$. We will call such channels emission channels.
%\begin{equation}
%    K^i = \sum_{j=1}^{N+1-i} L_{ij} \op{j}{j + i - 1}
%\end{equation}
%Resolution of identity imposes that 
%%\rd{(Details from comment to appendix!)}
%%\begin{align}
%%    \ev{\sum_{i=1}^N K_i^T K_i}{n}  & = \sum_{i=1}^N\sum_{j=1}^{N+1-i}\sum_{k=1}^{N+1-i} T_{ij}T_{ik} \ev{\op{k + i - 1}{k}\op{j}{j + i - 1}}{n} \\
%%                            & = \sum_{i=1}^N\sum_{j=1}^{N+1-i} (T_{ij})^2 \ip{n}{j + i - 1} = \sum_{i=1}^N (T_{i,\,n + 1 - i})^2 = 1
%%\end{align}
%%\begin{align}
%%	\sum_{i=1}^N (L_{i,\,n + 1 - i})^2 = 1\,,
%%\end{align}
%which means that there can be at most one 1 on each antidiagonal of the matrix $L$.

The entropy for any emission channel represented by its matrix $L$ is given by
\begin{equation}
    S^{map}(\Phi) = \frac{1}{N}\qty(N \log N - 
    \sum_{i=1}^N d_i \log d_i) \label{Lchanen}
\end{equation}
where $d_i = \sum_j L_{ij}$ is the number of ones in rows. Similarly, the entropy of the complementary channel
$\tilde\Phi$ can be given in terms of the numbers of ones $\tilde d_i =
\sum_j L_{ji}$ in columns of matrix $L$, 
\begin{equation}
    S^{map}(\tilde\Phi) = \frac{1}{N}\qty(N \log N - 
    \sum_{i=1}^{N}\tilde d_i \log \tilde d_i) \label{LTchanen}\,.
\end{equation}
Thus, the entropies may be calculated from the number of ones in
columns and rows of the $L$ matrix straight away. We can see that, in
terms of entropy, such channels may be denoted 
simply by an ordered pair
of unordered sets $(\qty{d_i},\qty{\tilde d_j})$. 

To illustrate the point, we show below these triangular matrices for qubit one-step emission channel $\Phi_{E}$  and identity channel $\Phi_\mathbb{I}$:

\begin{align*}
	L_\mathbb{E} = \mqty(1 & 0\\ 1 &), && L_{I} = \mqty(1 & 1 \\ 0&).
\end{align*}
Points in the entropy plane for these two channels are easily calculated from \eqref{Lchanen} and \eqref{LTchanen} as $(0,\log 2)$ and $(\log 2, 0)$, respectively. Interpolation between them can be denoted by matrix $A\qty(L_\mathbb{I},L_{E};x)$. Precise definition of channel $A$ is given in appendix \ref{LAdetails}

\subsection{Boundary of the allowed set for $N=3$}
In the case of qutrits, the next easiest system to consider, there are four important points in the boundary of $\mathcal{A}_3$, out of which three are similar to qubit channels already mentioned:

\begin{enumerate}
	\item $(0, \log 3)$: Unitary channels $\Phi_U$, exemplified by identity channel $K_1 = \mathbb{I}$, with matrix $L_U$
	$$
	L_U = \mqty(1 & 1 & 1 \\ 0 & 0 \\ 0);
	$$
	\item $(\log 3, 0)$: Emision channel $\Phi_{E}$, with Kraus channels $K_i = \op{1}{i}$ for $i = 1, 2, 3$, defined by a matrix $L_E$
	$$
	L_E = \mqty(1 & 0 & 0 \\ 1 & 0 \\ 1);
	$$
	\item $(\log 3, \log 3)$: Coarse graining channel $\Phi_{CG}$, with Kraus operators $K_i = \op{i}$ for $i = 1, 2, 3$;
	\item $(\frac{1}{3}\log(\frac{27}{4}),\frac{1}{3}\log(\frac{27}{4}))$: Partial spontaneous emission channels, which can be exemplified by a channel $\Phi_4$ given by its matrix $L_4$:

\begin{equation} \label{Lmatrix_channel4}
	L_4 = \mqty(1 & 0 & 1 \\ 1 & 0 \\ 0).
\end{equation}
\end{enumerate}

With these four channels we can characterize entire boundary for
qutrit channels with qutrit environment. Upper limits are analogous to
ones for qubit-qubit channels: $(\log 3, \Tilde{S})$ and $(S, \log 3)$
are trivially found by interpolating between coarse graining channel
$\Phi_{CG}$ and either unitary channel $\Phi_U$ or spontaneous
emission $\Phi_{SE}$, respectively. The lower boundary is conjectured
below: 
\begin{conj} \label{qutrit_bound_conj}
	For qutrit channels with qutrit environment the lower boundary
        of the allowed set $\mathcal{A}_3$ is given by the parametric curve
	\begin{align} \label{qutrit_boundary_eqn}
		p(a) = (-a\log a - (1-a) \log(1-a),\nonumber \\ \frac{\log 3}{3}-(a+\frac{1}{3})\log(a+\frac{1}{3}) - (a-\frac{1}{3})\log(a-\frac{1}{3}))	
	\end{align}
	with $a\in(0,\frac{1}{3})$, together with reflection through line $S = \Tilde{S}$, which contains all selfcomplementary channels $\Phi = \tilde\Phi$.
\end{conj}

\begin{figure}[H]
    \centering
    \includegraphics[width=.9\linewidth]{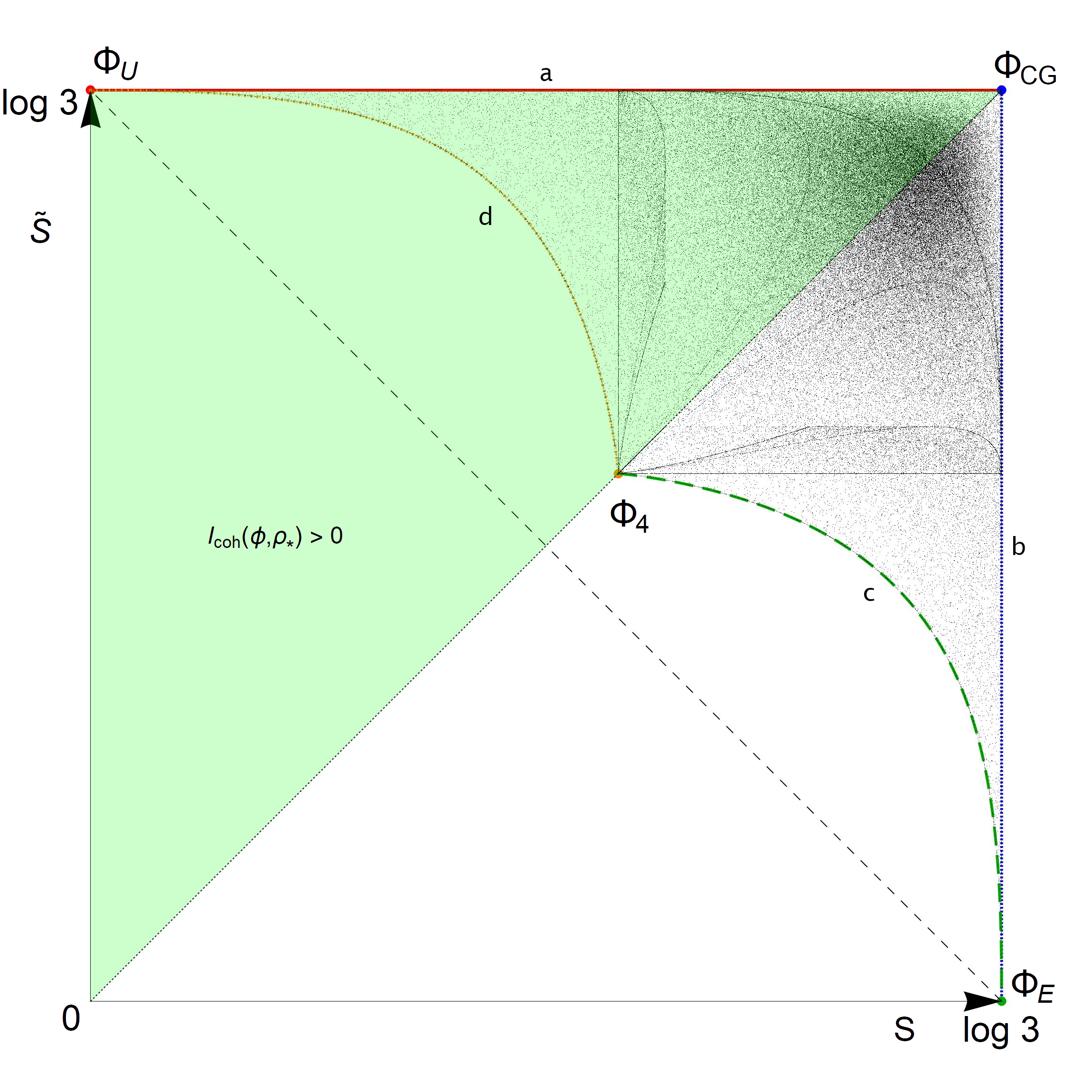}
    \caption{Scatter plot of randomly generated one qutrit channels of rank three
    which belong to the allowed set  $\mathcal{A}_3$
      in the entropy plane. Lines $a$ and $b$ saturate upper bounds on entropy of channel and its complementary, respectively, and are easily found by interpolating between coarse graining channel $\Phi_{CG}$ and spontaneous emission channel $\Phi_{SE}$ or unitary channel $\Phi_U$. Minimal curve $c$, described by parametric formula \eqref{qutrit_boundary_eqn}, can be found by interpolation between $\Phi_U$ and $\Phi_4$, defined by matrix $L$ in eqn. \eqref{Lmatrix_channel4}.}
    \label{fig:qtritfullbound}
\end{figure}

The parametric form of the curve $d$ described shown in figure \ref{fig:qbitfullbound} can be obtained by
considering entropy of the channels found from matrix
$A(L_\mathbb{I},L_{\Phi_4};1-3a)$, as defined in Appendix
\ref{LAdetails}, which can be given explicitly in terms of two Kraus operators

\begin{align*}
	K_1 = \mqty(
			1 & 0 & 0 \\
			0 & \sqrt{1 - 3a} & 0 \\
			0 & 0 & 1
		), &&
	K_2 = \mqty(
			0 & \sqrt{3 a} & 0 \\
			0 & 0 & 0 \\
			0 & 0 & 0		
		).
\end{align*}
The curve $c$ can be found in similar manner by considering interpolation by $A(L_E, L_{\Phi_4}; 1 - 3a)$. %\qq

As we were unable to find any channels below the
curve described above by extensive numerical probing of 
qutrit-qutrit
channels by method described in Appendix \ref{Channel_gen}, we believe
the conjecture should hold.
Further checks through evolution via random Hamiltonians have been outlined in Appendix \ref{hamiltonianEvo}, giving further evidence to the conjecture.

\section{Bound saturation in product dimensions}

Let us consider an especially compelling example of a matrix $L$ of size 4: 
%\qq It's better to write equations!\qq I.e. $T=...$ in all these examples.
\begin{equation*}
	L = \mqty(1 & 0 & 1 & 0 \\ 1 & 0 & 1 \\ 0 & 0 \\ 0)
\end{equation*}
for which $\qty{d_i} = \qty{\tilde d_i} = \qty{2,2}$. Using equations \eqref{Lchanen} and \eqref{LTchanen} we find that the total entropy of
such channel and its complementary yields $S\qty(\Phi) +
S\qty(\tilde\Phi) = \log 4$, which saturates the bound
\eqref{prop2ieq3}. This leads us to a more general
statement: 

\begin{prop}
	Consider a state of dimension $N = N_A N_B$ and a unitary matrix $U$ of the same dimension. We define the channel $\Phi$ by its Kraus operators:
	\begin{equation}
		\qty(K_\alpha)_{ij} = \begin{cases}
		 U_{i + (\alpha -1)N_A,\,j} & {\rm for}\, i = 1,\hdots,N_A \\
		 0 & {\rm for}\,i > N_A
		\end{cases}.
	\end{equation}
	With $\alpha = 1,\hdots,N_B$ and $%\qq i,\,
        j = %\qq N
        1,\hdots, N$ For such channels the bound \eqref{prop2ieq3} is saturated.
\end{prop}

\begin{proof}
	Using equation \eqref{prop1eq1} we consider the action of channel on the maximally mixed state, given by
	\begin{equation}
		\Phi(\rho_*) = \sum_{\alpha=1}^{N_B} K_\alpha \frac{\mathbb{I}_N}{N} K_\alpha^\dag = \frac{N_B}{N} \sum_{j=1}^{N_A}\op{j}
	\end{equation}
	which gives entropy $S(\Phi(\rho_*)) = S^{map}(\tilde{\Phi}) = \log N_A$. Analogously, using relation \eqref{Krausrowswap} for the Kraus operators of the complementary channel, we find that $S^{map}(\Phi) = \log N_B$, which completes the proof.
\end{proof}

Noting that this family of channels is highly general, we formulate following conjecture

\begin{conj}
  The number of distinct %\qqq number
  channels saturating the bound \eqref{prop2ieq3} in the allowed set $\mathcal{A}_N$ is equal to the number of divisors of $N$, including $1$ and $N$ itself.
\end{conj}

%\rd{The proof generalizes easily to $n$-partite systems in product dimensions and channels acting on them.
%
%\begin{prop}
%	The channels of the form
%	\begin{equation}
%		\Phi(\rho) = \bigotimes_{i=1}^n G_i(\rho_i)
%	\end{equation}
%	where $G_i$ is either a unitary channel $U_i$ acting on the
%        subspace of dimension $N_i$ or 
%        the 
%        channel complementary to unitary $\Tilde{U}_i$, saturate the bound \eqref{prop2ieq3} for $N = \prod_{i=1}^n N_i$. 
%%        \qq What do you
%%        mean with $\lor$? Concatenation?
%\end{prop}
%\begin{proof}
%	The proof follows from the fact that for every $U_i$ we have
%        a $\log(N_i)$ contribution in
%        $S^{map}\qty(\Tilde{\Phi})$. Likewise, for every $\Tilde{U}_i$
%        we get analogous contribution in
%        $S^{map}\qty(\Tilde{\Phi})$. Gathering all the contributions
%        together, the bound \eqref{prop2ieq3} is saturated. 
%\end{proof}}

%\rd{\textit{Here figures of full boundaries for both qubit and qutrit should be shown with short discussion of families leading to the boundaries. Hypothesis on full characterization of lower boundary curve for any dimension should be formulated as well as proof that projector-depolarising and projector-identity families saturate both upper bound for systems with $N=M$ is to be provided here as well}}

Furthermore, extensive numerical searches for dimension $N = 4, 5$
suggest the following conjecture on the boundary:
\begin{conj}
	For any dimension $N$ the entire lower boundary of available region
	  $\mathcal{A}_N$	 in
	the entropy plane $(S,\Tilde{S})$  can be found as a family of curves attained by channels generated from matrices $A(L_1,L_2;x)$.\\
\end{conj}

%\rd{
In Appendix \ref{qquartbounds} we provide the conjectured form of the boundary
curve restricting the allowed region $\mathcal{A}_4$ in the entropy plane
 for maps acting on four level systems 
  with the corresponding quantum channels.

%\textit{I'm not sure this should be mentioned, but for now I'm leaving it here}}

%\section{Summary}
%
%...

\section{Concluding remarks}

We established a general trade-off relation
concerning the operation entropy of a given channel $\Phi$
and its complementary $\tilde\Phi$.
In this way we bounded from below the sum of operation entropies, $S^{map}(\Phi)+ S^{map}(\tilde \Phi)$,
which characterizes the sum of information on the initial state $\rho$
accessible to the receiver of the output state $\rho'$ 
and the evesdropper controlling the environment.

Furthermore, we provided an exact characterisation of the boundary
of the allowed set $\mathcal{A}_2$
describing all single-qubit quantum channels
 in the entropy plane $(S^{map}(\Phi),S^{map}(\tilde \Phi))$.%\qq put back labels 'map' 
 Similar results concerning one-qutrit channels are formulated
 as a conjecture %\qq, which
 that allows one to predict the
  general form of the boundary of the analogous sets
  $\mathcal{A}_N$   in the entropy plane
 for quantum operations acting on a system of arbitrary dimension $N$. 
 In this way we identified a particular
class of quantum channels which minimize
the entropy  ${S}^{map}(\tilde\Phi) $ of the complementary channel 
among all the channels with fixed map entropy  $S^{map}(\Phi)$.
%\qq
 Our results may find applications in quantum thermodynamics in situations where one wants to minimize entropy production in both a system and its environment. 
% \qq
\bigskip

%{\sl more arguments concerning relevance of the results obtained and their possible applications welcome!}

\bigskip

\medskip

{\em Acknowledgments:} It is a pleasure to thank  Wojciech Roga and {\L}ukasz Rudnicki for fruitful discussions. %\qq

Financial support by the Polish National Science
Centre (NCN) under the grant number DEC-2015/18/A/ST2/00274 %(K.{\.Z}) 
is acknowledged.

\begin{widetext}

\appendix

%\section{Appendix}
%
%\setcounter{equation}{0}
%\setcounter{figure}{0}
%\renewcommand{\theequation}{A\arabic{equation}}
%\renewcommand{\thefigure}{A\arabic{figure}}

\section{Proof of Proposition 1 and Observation 1} \label{proofsofprops}

\begin{proof}
%   {\bf To be shifted to Appendix}
    Let us consider an extended system $ABC$ with subsystems $A$ and $B$ of dimensionality $N$ and subsystem $C$ of dimensionality $M$ in a pure state
    \begin{equation}
        \rho_{ABC} = \op{\psi_+}_{AB}\otimes\op{1}_C
    \end{equation}
    where $\ket{\psi_+} = \sum_{i=1}^N 1/\sqrt{N} \ket{i}\otimes\ket{i}$ is the $N$-dimensional Bell state.
    Assume now that the operation $\Phi$
    is induced by a unitary operation $U_{BC}$, which couples the system $B$
    with $M$-dimensional environment $C$ initially in a pure state,
    $   \Phi(\rho) = \Tr_C U_{BC}\qty(\rho_* \otimes \op{1})U_{BC}^\dag$.
  Then  we can introduce a tri-partite unitary operation,
   $W = 1_A \otimes U_{BC}$,
    and act with it on state the $\rho_{ABC}$,  obtaining another pure state,
    \begin{equation}
        \rho'_{ABC} = W\rho_{ABC}W^\dag. 
        \label{evstate}
    \end{equation}
    For any bipartite pure state the entropies of both partial traces are equal,  
    \begin{equation}
        S\qty(\rho'_{AB})  = S\qty(\rho'_C) %\label{proofeq1}
        \ \ {\rm and} \ \
        S\qty(\rho'_{AC})  = S\qty(\rho'_B) \label{proofeq2},
    \end{equation}
    which with use of the definition of the  complementary channel $\Tilde\Phi$ 
     implies the desired relations \eqref{prop1eq1}
\end{proof}

\begin{proof}
    Let us again consider the state $\rho'_{ABC}$ from \eqref{evstate} an its reductions again. To prove inequality \eqref{prop2ieq3} we will consider Araki-Lieb inequality
    
    \begin{equation}
        \abs{S\qty(\rho'_A) - S\qty(\rho'_B)} \leq S\qty(\rho'_{AB}). \label{ArakiLieb}
    \end{equation}
    Subsystem $A$ does not change under the action of the unitary operation $W$, so that we have $S\qty(\rho'_A) = S\qty(\rho_*) = \log N$. Furthermore, the dimensions of systems $A$ and $B$ are the same, so that $S(\rho'_B) \leq S(\rho'_A)$ and $\abs{S\qty(\rho'_A) - S\qty(\rho'_B)} = S\qty(\rho'_A) - S\qty(\rho'_B)$. Using \eqref{proofeq2} we can see that 
    \begin{equation}
        \log N \leq S\qty(\rho'_{AB}) + S\qty(\rho'_{AC}) \label{proofieq3}
    \end{equation}
    which proves inequality \eqref{prop2ieq3}.
    
    % Finally we consider the subadditivity conditions on subsystems $AB$ and $AC$
    % \begin{align}
    %     S\qty(\rho'_{AB}) & \leq S\qty(\rho'_A) + S\qty(\rho'_B) \label{subad1} \\
    %     S\qty(\rho'_{AC}) & \leq S\qty(\rho'_A) + S\qty(\rho'_C) \label{subad2}
    % \end{align}
    % which together with \eqref{proofeq1} and \eqref{proofeq2} yield
    
    % \begin{equation}
    %     \abs{S\qty(\rho'_{AB}) - S\qty(\rho'_{AC})} \leq \log N
    % \end{equation}
    % which proves the last inequality \eqref{prop2ieq4}
\end{proof}

\section{Channels determined by matrix $L$ and $A$} \label{LAdetails}

For any matrix $L$ the block structure of the corresponding Choi-Jamio{\l}kowski state
$\sigma_\Phi$ is easily given in terms of the entries  
of the 
matrix $L$,
%\rd{(Details from comment to appendix!)}
%\begin{align}
%    \qty(\frac{D_\Phi}{N})_{\mu\nu}     & = \sum_{i=1}^N K_i \op{\mu}{\nu} K_i^T \nonumber\\
%                                        & = \sum_{i=1}^N\sum_{j=1}^{N+1-i}\sum_{k=1}^{N+1-i} T_{ij}T_{ik}\op{j}{j+i-1}\op{\mu}{\nu}\op{k+i-1}{k} \nonumber \\
%                                        & = \sum_{i=1}^N\sum_{j=1}^{N+1-i}\sum_{k=1}^{N+1-i} T_{ij}T_{ik} \delta^\mu_{j+i-1} \delta^\nu_{k+i-1} \op{j}{k} \nonumber \\
%                                        & = \sum_{i=1}^N T_{i\,\mu + 1 - i} T_{i\,\nu + 1 - i} \op{\mu + 1 - i}{\nu + 1 - i}.
%\end{align}
\begin{align}
    \qty(\sigma_\Phi)_{\mu\nu} = \sum_{i=1}^N L_{i\,\mu + 1 - i} L_{i\,\nu + 1 - i} \op{\mu + 1 - i}{\nu + 1 - i}
\end{align}
with $L_{ij}\in\qty{0,1}$, $\mu,\nu = 1,\hdots,N$ and all projectors $\op{i}{j}$ referring to the system on which the channel $\Phi$ acts, without the environment.
We notice that blocks of the 
Choi-Jamio{\l}kowski
state $\sigma_\Phi$ either remain as rank one
projectors under summation or they are zeroed out. In general every
such channel generates a ``double-block'' structure of $\sigma_\Phi$ where the
size of corresponding blocks can be read out from the number of ones
in consecutive columns.  

To illustrate the point, let us take a particular matrix $L$ and
construct the corresponding Choi-Jamio{\l}kowski state. 
\begin{align}
   L =  \mqty(1 & 0 & 0 \\ 1 & 1 \\ 0) && \rightarrow && \sigma_\Phi & = \frac{1}{3}\mqty(
	\op{1} & 0 & 0 \\
	0 & \op{1} & \op{1}{2} \\
	0 & \op{2}{1} & \op{2}   
   )
\end{align}

To further simplify notation, we introduce a notion of interpolation between two channels represented by matrices $L_1$ and $L_2$. The matrix $A$ is introduced in terms of its entries:
$\qty(A(L_1,L_2;x))_{ij} = \sqrt{x \qty(L_1)_{ij} + \qty(1-x)\qty(L_2)_{ij}}$. In order to reconstruct the corresponding Kraus operators, we use the same formula as for matrices $L$ 

\begin{equation*}
	K_i = \sum_{j=1}^{N+1-i} A_{ij} \op{j}{j + i - 1}.
\end{equation*}
% \qq and
It can be easily found that the Kraus operators obtained in this way fulfil the identity resolution if $\sum_{i=1}^N \qty(A_{i,\,N + 1 - i})^2 = 1$. 

As an example let us consider interpolation between channels given in terms of $L$ matrices as

\begin{align*}
	L_1 & = \mqty(1 & 0 & 0 \\ 1 & 0 \\ 1) & L_2 & = \mqty(1 & 0 & 0 \\ 1 & 1 \\ 0)
\end{align*}
which give $A$ matrix and corresponding Choi-Jamio{\l}kowski states:

\begin{align}
   A(L_1,L_2;x) = \mqty(1 & 0 & 0 \\ 1 & \sqrt{1-x} \\ \sqrt{x}) 
                        && \rightarrow && \sigma_\Phi & = 
    \frac{1}{3}\mqty(
		\op{1} & 0 & 0 \\
		0 & \op{1} & \sqrt{1 - x}\op{1}{2} \\
		0 & \sqrt{1 - x}\op{2}{1} & x\op{1} + (1-x)\op{2}   	
    	) \nonumber\\
                        && \rightarrow && \sigma_{\Tilde{\Phi}} & = 
    \frac{1}{3}\mqty(
    	\op{1} & \op{1}{2} & \sqrt{x}\op{1}{3} \\
		\op{2}{1} & \op{2} & \sqrt{x}\op{2}{3} \\
		\sqrt{x}\op{3}{1} & \sqrt{x}\op{3}{2} & x\op{3} + (1-x)\op{2}
		)
\end{align}
Similar notion of upper-triangular matrix $A$ can be extended to arbitrary entries as long, as they fulfil the condition imposed by the identity resolution, $\sum_{i=1}^N \qty(A_{i,\,N + 1 - i})^2 = 1$.

\section{Generation of the channels from the allowed set $\mathcal{A}_N$ in the entropy plane} 
\label{Channel_gen}

%\rd{\textit{Explanation of method for generating the unitaries of all possible types - random block diagonal of unitaries of dimensions summing to $N$ + random permutations $P_1$ and $P_2$. Perhaps the code for generation of unitaries should be provided?}}

Standard generation of Kraus operators for $N$ dimensional system with $M$ dimensional environment involves generating unitary matrix $U$ of dimension $N\cdot M$ with respect to flat measure and assuming that its first block column corresponds to the set of Kraus operators, that is

\begin{equation}
    K^i_{jk} = U_{j+\qty(i-1)N\,k}
\end{equation}

This, however, would yield underrepresentation of channels in certain regimes of entropy. In order to overcome this issue we propose different way of generating unitary matrices, described below in steps.

\begin{enumerate}
    \item From all $J$ sets of positive integers $\qty{n_i}_{i=1}^j$ such that their sum is
      equal to the desired dimension of the unitary matrix, $\sum_i n_i =
      N\cdot M$,
    we take one at random with probability\nobreak~$1/J$.
    \item We generate set of unitary matrices of dimensions defined by the chosen set of integers 
    $\qty{U_i:\,\text{dim}(U_i) = n_i}_{i=1}^j$ 
    and construct a block-diagonal matrix
    
    \begin{equation}
        U = \mqty(\dmat[0]{U_1, U_2,\ddots,U_j})
    \end{equation}
    where zeroes are to be understood as matrices of dimension $n_i \cross n_j$ filled with zeroes.
    \item From all possible permutation operations of size $N\cdot M$
      we take two permutations $P_1$ and $P_2$ and define the unitary matrix
    \begin{equation}
        W = P_1 U P_2\,.
    \end{equation}
\end{enumerate}

The unitary operation $W$ allows us to probe the possible entropies of channels more uniformly than in the case of standard generation of Kraus operators.

\section{Proof of proposition 2} \label{prop3proof}

We start from the representation of the qubit channel with qubit environment by Kraus operators. By utilizing unitary freedom of pre- and postpreparation given by $K_i \rightarrow U K_i V$, we may transform first Kraus operator $K_1$ into a diagonal form with $a, b \in \mathbb{R}$ with use of standard procedure of singular value decomposition, thus getting Kraus operators of the form 

\begin{equation}
    K_1 = \mqty(a & 0 \\ 0 & b), \quad
    K_2 = \mqty(\alpha & \beta \\ \gamma & \delta). \label{kraus1form}
\end{equation}
Additional conditions imposed by decomposition of unity $\sum_i K_i^\dag K_i = \mathbb{I}$ can be rewritten as

\begin{align}
    a^2 + \abs{\alpha}^2 + \abs{\gamma}^2 & = 1,\label{udeco1}\\
    b^2 + \abs{\beta}^2 + \abs{\delta}^2 & = 1,\label{udeco2} \\
    \alpha\beta^* + \gamma\delta^* & = 0. \label{udeco3}
\end{align}
Equations \eqref{udeco1} and \eqref{udeco2} allow us to introduce 
phased spherical coordinates:
\begin{align*}
    a &= \cos\theta_1, & b & = \cos\theta_2, \\
    \alpha & = \sin\theta_1\cos\phi_1e^{i \chi_{11}}, & 
    \beta & = \sin\theta_2\cos\phi_2e^{i \chi_{21}}, \\
    \gamma & = \sin\theta_1\sin\phi_1e^{i \chi_{12}}, & 
    \delta & = \sin\theta_2\sin\phi_2e^{i \chi_{22}}
\end{align*}
and in turn rewrite Eq. \eqref{kraus1form} as

\begin{align}
    K_1 & = \mqty(  \cos\theta_1 & 0\\
                    0 & \cos\theta_2) &
    K_2 & = \mqty(  \sin\theta_1\cos\phi_1e^{i \chi_{11}} &
                    \sin\theta_2\cos\phi_2e^{i \chi_{21}} \\
                    \sin\theta_1\sin\phi_1e^{i \chi_{12}} &
                    \sin\theta_2\sin\phi_2e^{i \chi_{22}}) \label{krauspolarform}
\end{align}
with $\theta_1, \theta_2 \in (0,\pi)$ and $\phi_1,\phi_2,\chi_{11},\chi_{12},\chi_{21},\chi_{22}\in(0,2\pi)$.

From here we may inquire about subfamilies satisfying Eq. \eqref{udeco3}. We will do this in steps, eliminating possible cases one by one.

First we consider family of channels with $\sin \theta_1 = 0$ given by Kraus operators of the form

\begin{align*}
    K_1 & = \mqty(  1 & 0\\
                    0 & \cos\theta_2), &
    K_2 & = \mqty(  0&
                    \sin\theta_2\cos\phi_2e^{i \chi_{21}} \\
                    0 &
                    \sin\theta_2\sin\phi_2e^{i \chi_{22}})
\end{align*}
with complementary channel defined accordingly

\begin{align*}
    \Tilde{K}_1 & = \mqty(  1 & 0\\
                    0 & \sin\theta_2\cos\phi_2e^{i \chi_{21}}), &
    \Tilde{K}_2 & = \mqty(  0&
                    \cos\theta_2 \\
                    0 &
                    \sin\theta_2\sin\phi_2e^{i \chi_{22}}).
\end{align*}

Channels from this family, acting on the maximally mixed state $\rho_*
= \mathbb{I}/2$ give: 

\begin{align*}
    \rho_K & = \frac{1}{2}\mqty(
                                 \cos ^2(\phi_2) \sin ^2(\theta_2)+1 & e^{i (\chi_{21}-\chi_{22})} \cos (\phi_2) \sin (\phi_2) \sin ^2(\theta_2) \\
                                 e^{-i (\chi_{21}-\chi_{22})} \cos (\phi_2) \sin (\phi_2) \sin ^2(\theta_2) &  1 - \cos ^2(\phi_2) \sin ^2(\theta_2)) \\
    \rho_{\Tilde{K}}  & = \frac{1}{2}\mqty(
                     \cos ^2(\theta_2)+1 &  e^{-i \chi_{22}} \cos (\theta_2) \sin (\phi_2) \sin (\theta_2) \\
                     e^{i \chi_{22}} \cos (\theta_2) \sin (\phi_2) \sin (\theta_2) &  1-\cos[2](\theta_2)
                    )
\end{align*}
The eigenvalues for these states read:

\begin{align*}
    \lambda_{K,1} & = \frac{1}{2} + \frac{1}{2}\cos (\phi_2) \sin ^2(\theta_2) \\ 
    \lambda_{K,2} & = \frac{1}{2} - \frac{1}{2}\cos (\phi_2) \sin ^2(\theta_2) \\
    \lambda_{\Tilde{K},1} & = \frac{1}{2} + \frac{1}{8} \sqrt{2 \cos \qty(4 \theta _2) \cos ^2\qty(\phi _2)+8 \cos \qty(2 \theta _2)-\cos \qty(2 \phi _2)+7} \\
    \lambda_{\Tilde{K},2} & = \frac{1}{2} - \frac{1}{8} \sqrt{2 \cos \qty(4 \theta _2) \cos ^2\qty(\phi _2)+8 \cos \qty(2 \theta _2)-\cos \qty(2 \phi _2)+7}
\end{align*}

In order to retrieve dependencies between the channel eigenvalues, we will set $\lambda_{K,1} = a \in [0,\frac{1}{2})$. Solving for $\phi_2$ we obtain two solutions

\begin{equation}
    \phi_2 = \pm\arccos(\frac{1 - 2a}{\sin^2\theta_2}).
\end{equation}
 This gives $\lambda_{K,1}$ in terms of one variable $\theta_2$ and a constant $a$.

In case of qubits, minimization of von Neumann entropy $S(\rho) = -\Tr(\rho\log\rho)$ can be easily shown to be equivalent to minimization of linear entropy $S_{lin} = 1 - \lambda_{\Tilde{K},1}^2 - \lambda_{\Tilde{K},2}^2$. 
Consider the differential of $S(\rho)$ with eigenvalues $(\lambda, 1-\lambda)$. We find that %\qq

\begin{equation} \label{S_diff}
	\dd{S} = -\qty(\dd(\lambda\log\lambda) + \dd\qty[\qty(1-\lambda)\log(1-\lambda)]) = \qty(\log(1-\lambda) - \log\lambda)\dd{\lambda}
\end{equation}
which is zero if $\lambda = \frac{1}{2}$ or $\dd\lambda = 0$. Similarly, in case of linear entropy we find that%\qq

\begin{equation} \label{P_diff}
	\dd{S_{lin}} = \qty(2 - 4\lambda)\dd{\lambda}
\end{equation}
which is found to be zero under the same circumstances. Let us now consider the linear entropy of aforementioned $\rho_{\Tilde{K}}$ given as %\qq

\begin{equation}
    S_{lin} = 1 - \frac{1}{8} \qty(\cos ^2(\theta_2) \qty(2 \sin ^2(\theta_2) \cos \qty(2 \arcsin(\frac{2a-1}{\sin^2\theta_2}))+\cos (2 \theta_2)+3)+4) \label{firsteq}
\end{equation}

In order to extremize this we need 

\begin{equation} \label{purity1der}
	\pdv{S_{lin}}{\theta_2} = 0
\end{equation}
which leads to solutions $\theta_2 = \pi/2$, $\theta_2 = \arccos(\pm\sqrt{2a})$ or $\theta_2 = \arccos(\pm\sqrt{2(1-a)})$.

When $\theta_2 = \pi/2$, the answer is independent from $a$ and lies above the expected extremal curve. For the remaining solutions, they are equivalent up to permutation of eigenvalues between Kraus operators for channel or its complementary. For this reason, without loss of generality, we consider only one of them.

For $\theta_2 = \arccos{\sqrt{2a}}$ the eigenvalues read:
$
    \lambda_{K,1}  =a,\,
    \lambda_{K,2}  = 1 - a,\,
    \lambda_{\Tilde{K},1}  = \frac{1}{2} - a,\,
    \lambda_{\Tilde{K},2}  = \frac{1}{2} + a,\,
$
and the solution is well defined for $a\in[0,\frac{1}{2})$. 
In this way we obtain a parametric form of the curve representinng the  
boundary channels in terms of
both entropies $S$ and $\tilde S$,
% channel and complementary channel entropy as:

\begin{align}
  S = S^{map}(\Phi) %\qq added 'map'
  & = -\qty(a\log(a) + \qty(1-a)\log(1-a)), \nonumber \\
    {\tilde S}=  S^{map}(\Tilde{\Phi}) & = -\qty(\qty(\frac{1}{2} - a)\log(\frac{1}{2} - a) + \qty(\frac{1}{2} + a)\log(\frac{1}{2} + a)).
\end{align}
This is the curve given in Eq. \eqref{boundary_qubit_eqn}. Analogous procedure, consisting in setting one eigenvalue for $\Tilde{\Phi}(\rho_*)$ equal to $a$, extracting $\phi_1(a)$ and calculating %\qq extremes
the extremas %\qq
for linear entropy, leads to solution when $\sin(\theta_2) = 0$. For this reason, we only conclude without presentation that it leads to analogous conclusions. %\qq

The next step is to consider the remaining case of \eqref{udeco3}, leading to nontrivial condition
\begin{equation}
    \cos\phi_1\cos\phi_2e^{i(\chi_{11}-\chi_{12})} + \sin\phi_1\sin\phi_2e^{i(\chi_{21}-\chi_{22})} = 0. \label{n-trivial}
\end{equation}
Its consideration  will be split into further substeps. 

If $\sin\phi_1 = 0$, the condition \eqref{n-trivial} is reduced to
$
    \pm\cos\phi_2e^{i(\chi_{11}-\chi_{12})} = 0
$
which is equivalent to
$
    \cos\phi_2 = 0.
$
Such statement reduces the Kraus operators to the form

\begin{align*}
    K_1 & = \mqty(  \cos\theta_1 & 0\\
                    0 & \cos\theta_2), &
    K_2 & = \mqty(  \pm\sin\theta_1e^{i \chi_{11}} &
                    0 \\
                    0 &
                    \pm\sin\theta_2e^{i \chi_{22}}).
\end{align*}

Such channel, however, does not change the maximally mixed state $\Phi(\rho_*) = \rho_*$, and so the entropy of its complementary $S^{map}(\Tilde{\Phi}) = \log 2$ is maximal  
and as such does not come into our interest. Similar reasoning applies
to the case when $\sin\phi_2 = 0$. 

Now, if none of the trigonometric functions in condition
\eqref{n-trivial} is zero, the condition can be rewritten in the form 
\begin{equation}
    \tan \phi_1 \tan \phi_2 = - e^{i(\chi_{11} - \chi_{12} - \chi_{21} + \chi_{22})}.
\end{equation}
Since $\tan \phi_1, \tan \phi_2$ are real-valued, the condition can be split into two separate ones:

\begin{align}
    e^{i(\chi_{11} - \chi_{12} - \chi_{21} + \chi_{22})}  & = \pm 1 \nonumber \\
  \text{and}\qquad
  \phi_1 & = \arctan(\frac{\mp1}{\tan \phi_2}) = \pm\phi_{2\,\text{mod}\pi} \mp \pi/2\,.
\end{align}
Without loss of generality let us consider only the case with upper sign and $\phi_2\in(0,\pi)$, which gives

\begin{align*}
    e^{i(\chi_{11} - \chi_{12} - \chi_{21} + \chi_{22})}  & = 1 \\
    {\rm and }\qquad\phi_1 = \phi_2 - \frac{\pi}{2}\,.
\end{align*}

This gives Kraus operators in the form
\begin{align*}
    K_1 & = \mqty(  \cos\theta_1 & 0\\
                    0 & \cos\theta_2) \\
    K_2 & = \mqty(  \sin\theta_1\sin\phi_2e^{i \chi_{11}} &
                    \sin\theta_2\cos\phi_2e^{i \chi_{21}} \\
                    -\sin\theta_1\cos\phi_2e^{i \chi_{12}} &
                    \sin\theta_2\sin\phi_2e^{i \chi_{22}})\,.
\end{align*}
For this channel the result
$\Phi(\rho_*)$  of action on the maximally mixed state
 has eigenvalues of the form 
\begin{align*}
    \lambda_{1,2} & = \frac{1}{2} \pm \frac{1}{4} \sqrt{\cos ^2(\phi_2) (\cos (2 \theta_1)-\cos (2 \theta_2))^2} \\
    & = \frac{1}{2} \pm \frac{1}{4} \cos(\phi_2) \qty(\cos(2 \theta_1)-\cos(2 \theta_2))
\end{align*}
where, without loss of generality, we drop the absolute value. Now, we may assume, once again, that $\lambda_1 = a \in [0,\frac{1}{2})$, which allows us to solve for $\phi_2$, which gives
\begin{equation}
    \phi_2 = \pm \arccos \qty(\frac{2 (2 a-1)}{\cos (2 \theta_1)-\cos (2 \theta_2)})\,.
\end{equation}
Given this, we consider the linear entropy for the complementary channel, which yields: %\qq
\begin{align*}
    S_{lin} = 1 & - \frac{1}{4} \big(-(1-2 a)^2 \frac{\cos (\Delta_\chi)+1}{\sin^2(\theta_1-\theta_2)}+(1-2 a)^2 \frac{\cos (\Delta_\chi)-1}{\sin^2(\theta_1+\theta_2)} \\
     & +8 (a-1) a+\sin (2 \theta_1) \sin (2 \theta_2) \cos (\Delta_\chi)+\cos (2 \theta_1) \cos (2 \theta_2)+5)
 \end{align*}

First we notice that the quantity is dependent only on the difference
of the phases, $\Delta_\chi = \chi_{11} - \chi_{22}$, which reduces
the effective number of free parameters. Considering the derivative
with respect to it we get  
\begin{align*}
    \pdv{S_{lin}}{\Delta_\chi} & = \frac{1}{4} \sin (\Delta_\chi)\qty((1-2 a)^2 \qty(\frac{1}{\sin ^2(\theta_1-\theta_2)}- \frac{1}{\sin^2(\theta_1+\theta_2)})
     - \sin (2 \theta_1) \sin (2 \theta_2) ) = 0\,.
\end{align*}
From the solution of this equation we find that
\begin{equation}
    \Delta_\chi = n\pi\,.
\end{equation}
First we consider only $\Delta_\chi = 0$. Next, we need to consider
derivative with respect to $\theta_1$, which turns out to be dependent
only on $\Delta_\theta = \theta_1 - \theta_2$,
\begin{equation}
    \pdv{S_{lin}}{\theta_1} = \frac{1}{4} \left(\frac{4 (1-2 a)^2}{ \tan (\theta_1-\theta_2) \sin ^2(\theta_1-\theta_2)}-2 \sin (2 (\theta_1-\theta_2))\right) = 0\,.
\end{equation}
The solutions of this
equation are analogous %\qq as for
to those of %\qq
\eqref{purity1der}, yielding the same
extremal entropy curve as earlier. 
For $\Delta_\chi = \pi$, the equation is slightly more complicated
\begin{equation}
    \pdv{S_{lin}}{\theta_1} = \frac{32 a^2-32 a+4 \cos (2 (\theta_1+\theta_2))-\cos (4 (\theta_1+\theta_2))+5}{\tan (\theta_1+\theta_2) \sin ^2(\theta_1+\theta_2)}  = 0
\end{equation}
but again yields the %\qq same
solutions analogoues to %\qq
as \eqref{purity1der}.

This completes the proof, as the consideration exhausts the set of
possible qubit-qubit channels.  \hfill\qed

\section{Hamiltonian evolution for qutrit boundary} \label{hamiltonianEvo}

In order to solidify the conjecture \ref{qutrit_bound_conj} concerning the boundary of allowed region $\mathcal{A}_3$, %\qq
we performed %\qq
numerical checks using evolution under random hamiltonians $H$ of dimension $9$ drawn from GUE. Let us consider the channel given by Kraus operators $K_1(0), K_2(0), K_3(0)$. We can define its evolution by considering the block column formed by all the Kraus operators and transforming it under unitary generated by hamiltonian

\begin{equation}
	\mqty(K_1(t) \\ K_2(t) \\ K_3(t)) = \exp(i H t) \mqty(K_1(0) \\ K_2(0) \\ K_3(0))
\end{equation}

In the following figures \ref{fig:qtritboundevo1} and \ref{fig:qtritboundevo2} results of numerical checks are presented. 
For ease of computation, %\qq
instead of von Neumann entropy $S$, linear entropy $S_{lin}$ was used in the computations. %\qq
%As has been already shown in \eqref{S_diff} and \eqref{P_diff}, the two are extremized under the same conditions, so the results give equivalent answers. \qqq In app D this was shown for qubits only! \qqq
In order to justify this, consider the differential of linear entropy in any dimension, given by %\qq

\begin{align*}
	\dd{S_{lin}} = \dd(1 - \sum_{i=1}^N \lambda_i^2) = -2\sum_{i=1}^{N-1} \lambda_i\dd{\lambda_i} - 2\qty(1-\sum_{j=1}^{N-1}\lambda_j)\dd(1-\sum_{j=1}^{N-1}\lambda_j) = -2\sum_{i=1}^{N-1} \qty(\lambda_i +\sum_{j=1}^{N-1}\lambda_j - 1)\dd{\lambda_i}
\end{align*}
which is zero when $\dd\lambda_i = 0$ for all $i$ or for $\lambda_i = \frac{1}{N}$, which can be calculated by direct solving of the system of equations of the form $\lambda_i +\sum_{j=1}^{N-1}\lambda_j - 1 = 0$.

Similarly, consider the differential of von Neumann entropy, that is given by

\begin{align*}
	\dd{S} & = -\dd(\sum_{i=1}^N \lambda_i \log \lambda_i) = - \sum_{i=1}^{N-1} \log\lambda_i \dd\lambda_i - \log(1 - \sum_{j=1}^{N-1}\lambda_j)\dd(1 - \sum_{j=1}^{N-1}\lambda_j) = \sum_{i=1}^{N-1}\log(\frac{1 - \sum_{j=1}^{N-1}\lambda_j}{\lambda_i})\dd\lambda_i.
\end{align*}
this is zero either when $\dd\lambda_i = 0$ for every $i$ or when the system of equations of the form $\frac{1 - \sum_{j=1}^{N-1}\lambda_j}{\lambda_i} = 1$ is satisfied. However, by elementary manipulation it is seen that they are equivalent to conditions for extrema of linear entropy. Thus, the equivalence is established. %\qq
% \qq By these assumptions %\qq what assumptions
In terms of visual appearance, the difference between the plots presented in Fig. \ref{fig:qtritboundevo1}, \ref{fig:qtritboundevo2} and the plot in Fig. \ref{fig:qtritfullbound} is %\qq apparent
a non-uniform rescaling and flip of both vertical and horizontal axes.

\begin{figure}[H]
    \centering
    \includegraphics[width=.45\linewidth]{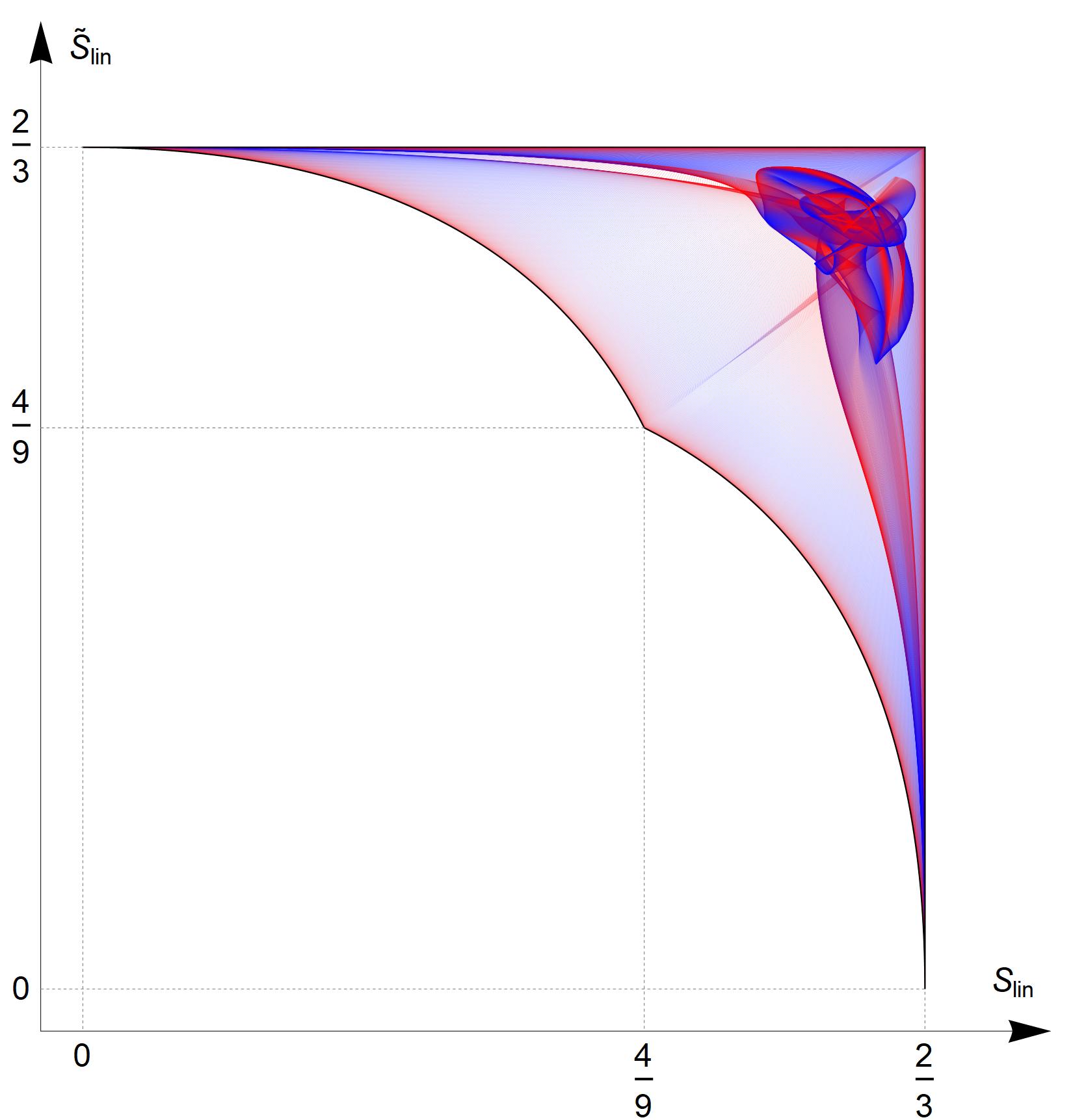}
    \includegraphics[width=.45\linewidth]{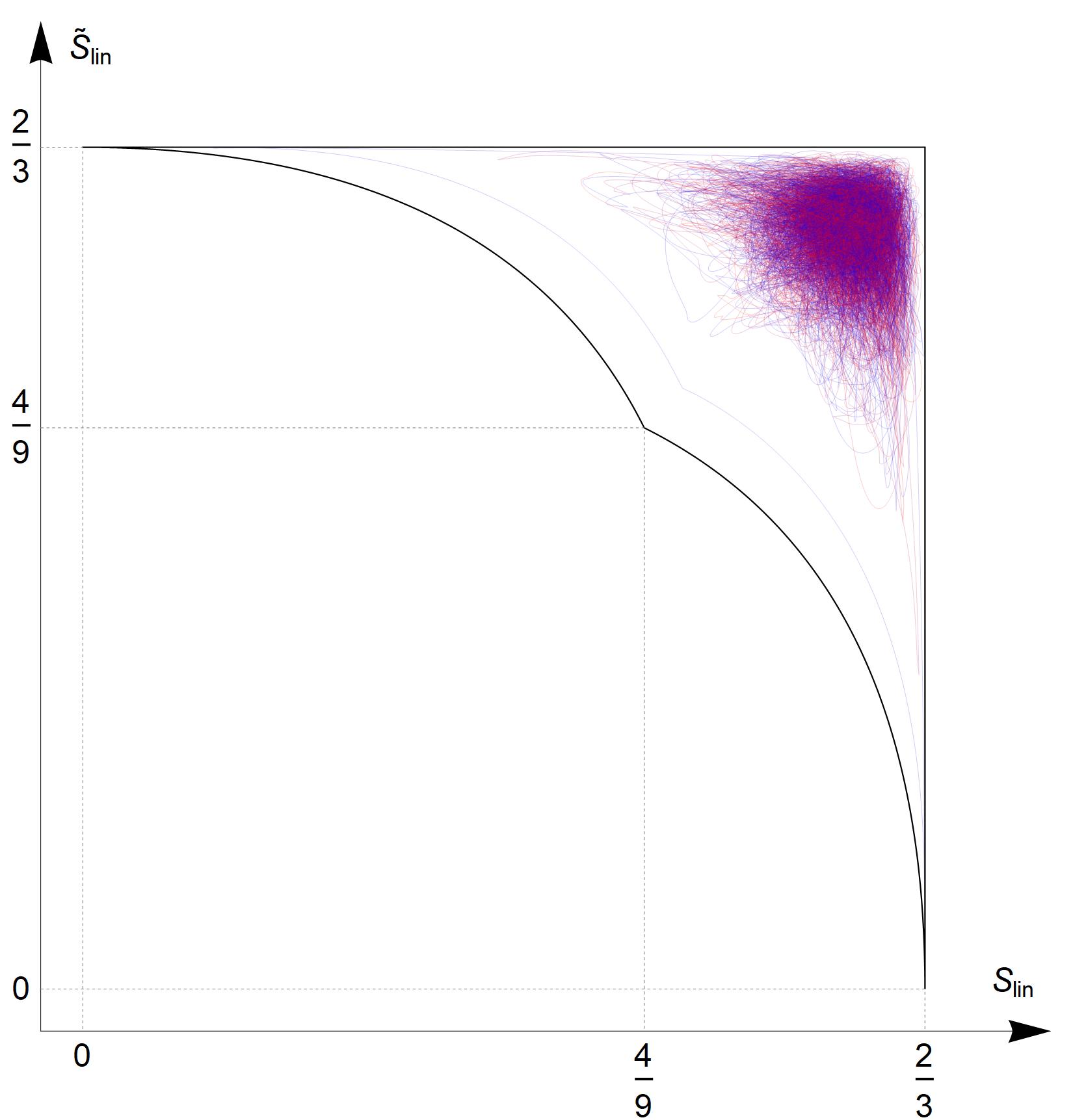}
    \caption{%\qq In the two figures above
      Evolution, described in Appendix \ref{hamiltonianEvo}, of linear entropies of channels drawn from the conjectured boundary of allowed region $\mathcal{A}_3$ for qutrit channels under fixed hamiltonian $H$ drawn from GUE. In both panels each %\qq
      point on a time slice %\qq in my understanding a time-slice is a slice of constant time. 
      is connected by a line with the corresponding point on the next time slice, and color-coding periodic in time
      is employed. In the left panel detailed evolution with $t\in[0,1]$ and timestep 0.001 is shown, while in the right panel we can see evolution for $t\in[0,100]$ with timestep 0.1. In neither %\qq
      of the two%\qq
       cases,       channels outside of the conjectured boundary, given in black, have been found.} 
    \label{fig:qtritboundevo1}
\end{figure}

\begin{figure}[H]
    \centering
    \includegraphics[width=\linewidth]{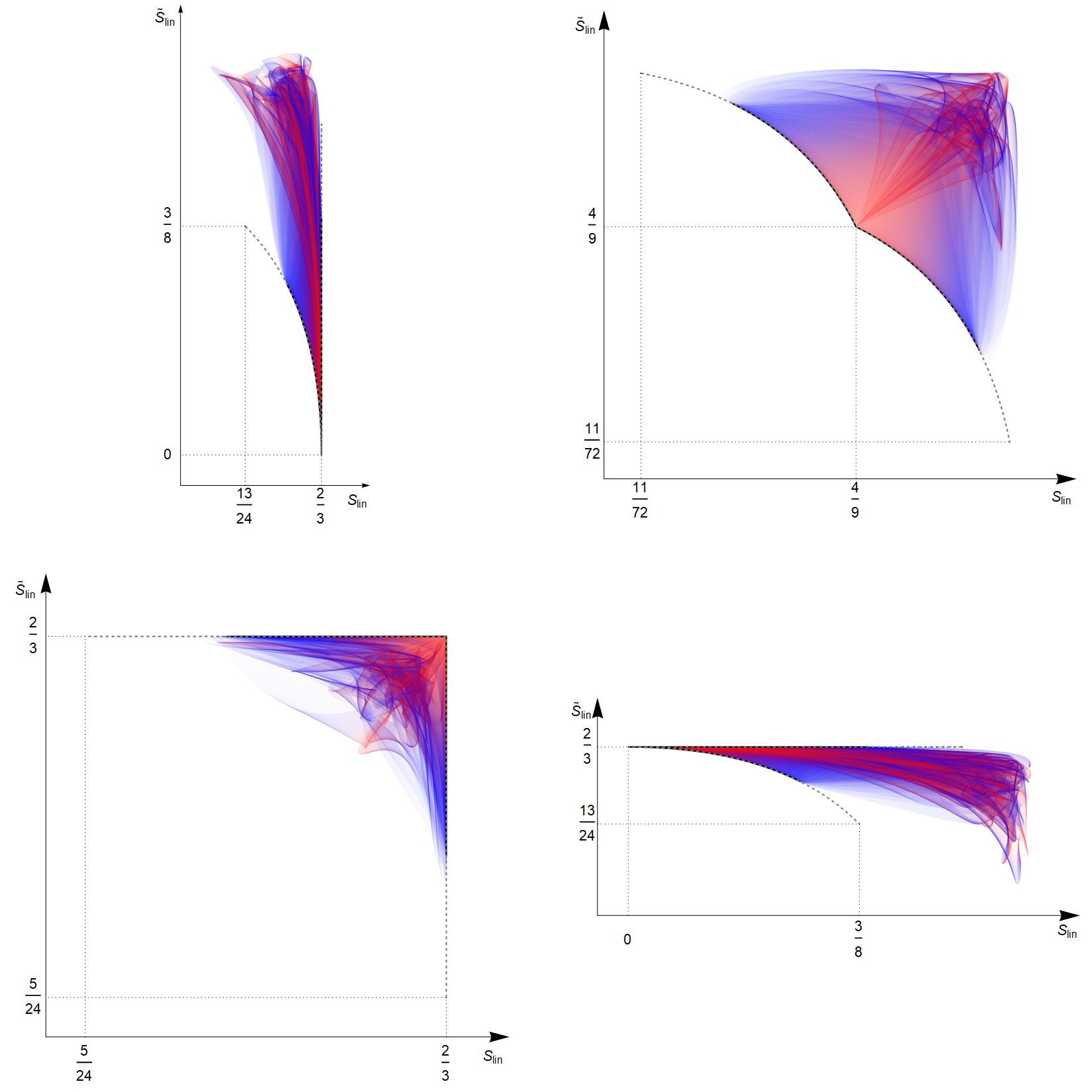}
    \caption{%\qq In the four panels above
      Evolution of linear entropies, described in Appendix \ref{hamiltonianEvo}, of channels taken from vicinity of cusp points in the boundary of allowed region $\mathcal{A}_3$.%\qqq has been shown.
      Color-coding has been employed here to distinguish evolutions originating from different points on the boundary. Black lines correspond to the region of origin of channels, whereas gray dashed extensions are further fragments of the boundary. Each of the panels includes evolutions with respect to 30 random hamiltonians for $t\in[0,0.5]$ with timestep 0.01. For none of the 30 hamiltonians evolution beyond the boundary has been found.} 
    \label{fig:qtritboundevo2}
\end{figure}

\section{Boundary for $N=4$} \label{qquartbounds}

In order to obtain the full lower boundary of the allowed set $\mathcal{A}_4$ we define three emission channels by their matrices $L$.

\begin{align*}
L_1 & = \left(
\begin{array}{cccc}
 1 & 0 & 0 & 0 \\
 1 & 0 & 0 \\
 1 & 0 \\
 1 \\
\end{array}
\right), &
L_2 & = \left(
\begin{array}{cccc}
 1 & 0 & 0 & 0 \\
 1 & 0 & 0 \\
 1 & 1 \\
 0  \\
\end{array}
\right), &
L_3 & = 
\left(
\begin{array}{cccc}
 1 & 1 & 0 & 0 \\
 0 & 0 & 0 \\
 1 & 1 \\
 0 \\
\end{array}
\right),
\end{align*}
among which $L_1$ is in fact the identity channel $\Phi(\rho) = \rho$.

Using these we define two interpolation channels as described in Appendix \ref{LAdetails}:

\begin{enumerate}
	\item $A(L_1,L_2;x)$ with $x = 1 - 4a$, giving a parametric curve
$$
p(a) = (-a\log a - (1-a) \log(1-a), \log 2-(a+\frac{1}{4})\log(a+\frac{1}{4}) - (a-\frac{1}{4})\log(a-\frac{1}{4}))	
$$ for $a\in(0,\frac{1}{4})$.
	\item $A(L_2,L_3;x)$ with $x = 4a - 2$, giving a parametric curve
$$
p(a) = (\frac{\log 4}{4} -(1-a)\log(1-a) - (a-\frac{1}{2})\log(a-\frac{1}{2}), -a\log a - (1-a) \log(1-a))	
$$ for $a\in(\frac{1}{2},\frac{3}{4})$.
\end{enumerate}

Remaining part of the boundary may be given as reflection through the line $S = \Tilde{S}$, containing the self-complementary channels $\Phi = \tilde \Phi$.
%\rd{DISCLAIMER: a lot of difficult algebraic manipulation was computer-handled, so most probably not all equations are in the easiest forms possible. They can be simplified as needed. The solutions, however, are precise.}

\section{Boundaries for qubit and qutrit systems with maximally extended boundaries}

To further extend the analysis, in the following Fig. \ref{fig:fullenvs} we provide full boundaries for qubit and qutrit systems with environment extended to dimension $N^2$ in order to cover all the available region in the entropy plane. The results given in this paper allow us to form more precise boundaries than in \cite{RPRZ13}. Moreover, in case of qubits the bounds in Prop. \ref{prop_qbit_bound} are proven to be tight.

\begin{figure}[H]
    \centering
    \includegraphics[width=.49\linewidth]{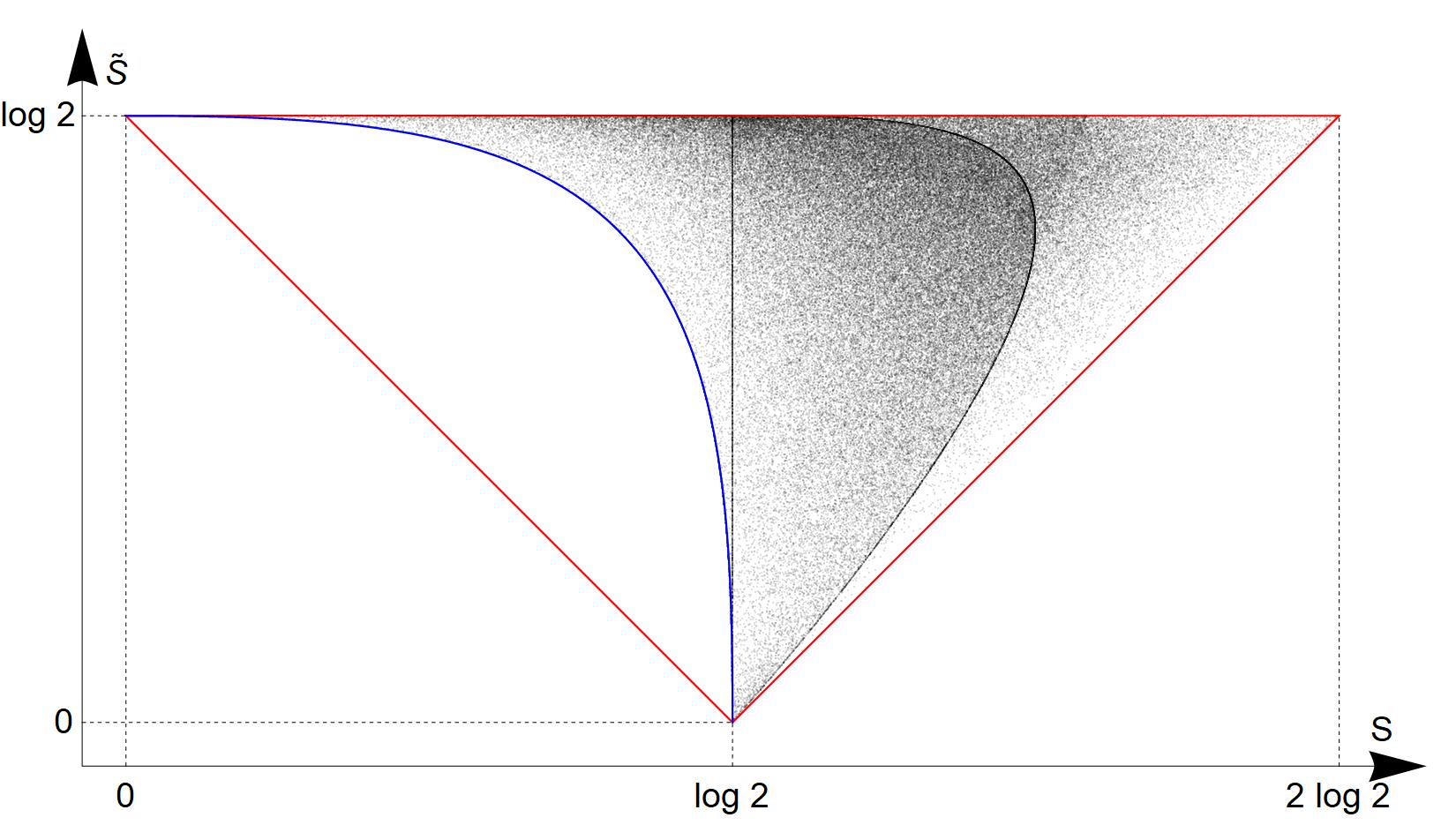}
    \includegraphics[width=.49\linewidth]{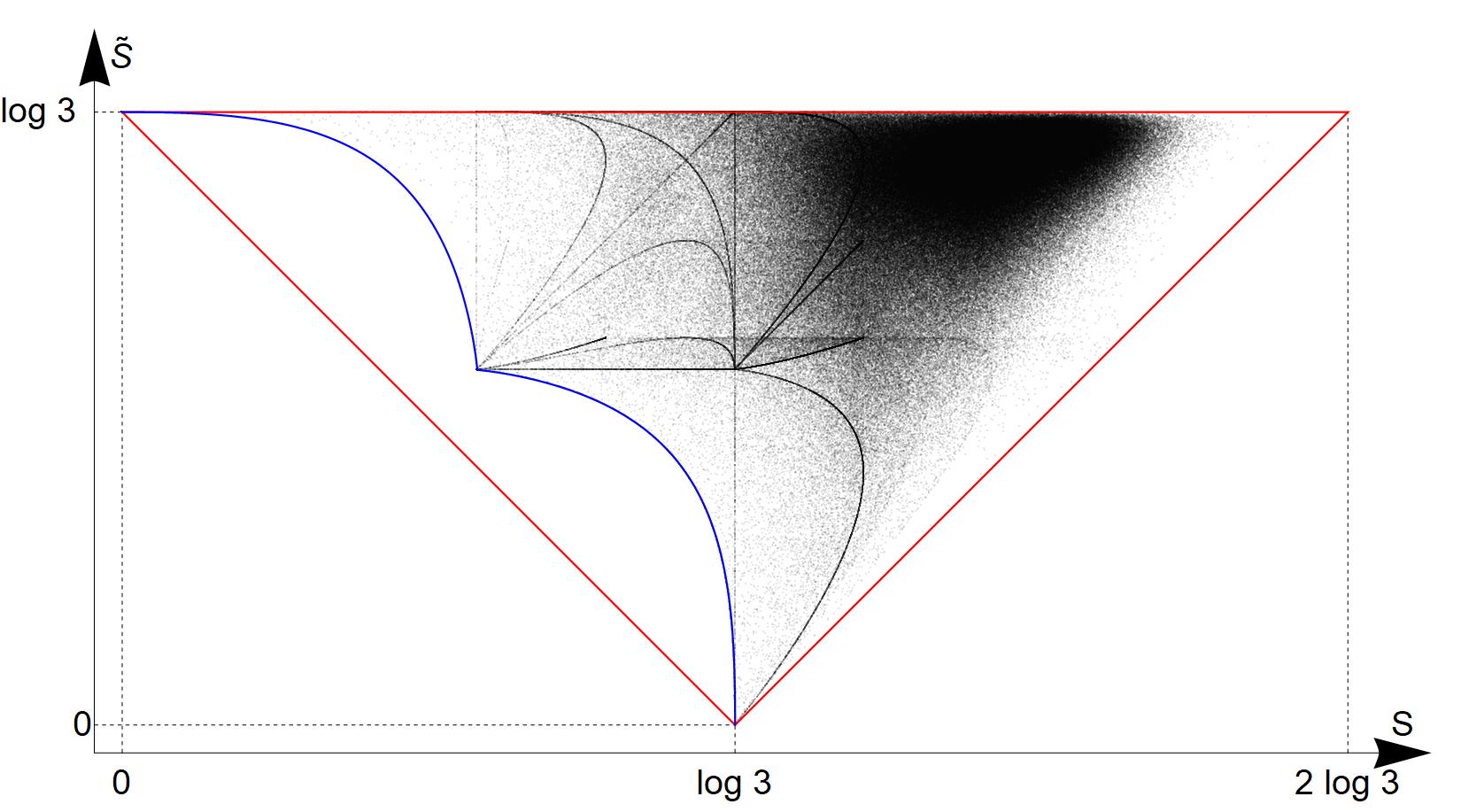}
    \caption{%\qqq In the two panels above we plot
      Possible entropies of channels and their %\qq entropies,
      complementaries for qubits and qutrits in the respective
      panels. Black dots are channels generated according to the
      method given in Appendix \ref{Channel_gen}. Red triangles are
      the boundaries obtained from Obs.~\ref{obs1}, Prop. \ref{prop1}
      and subadditivity for von Neuman entropy. Blue lines are the
      proposed tight lower bounds, proved for qubits in
      Prop.~\ref{prop_qbit_bound} and conjectured for qutrits in
      Conj.~\ref{qutrit_bound_conj}. As channels corresponding to the
      vertices are well known, the empty region on the right side of
      qutrit plot should be regarded as an artifact of the chosen
      method of channel generation.} 
    \label{fig:fullenvs}
\end{figure}

\end{widetext}

\end{document}